\documentclass[a4paper, onecolumn, 11pt]{quantumarticle} 
\pdfoutput = 1
\usepackage[utf8]{inputenc}
\usepackage[english]{babel}
\usepackage[T1]{fontenc}
\usepackage{amsmath, amssymb,amsthm}

\usepackage{tikz}
\usepackage{pgfplots}
\pgfplotsset{compat=1.18}
\usepackage{algorithm}
\usepackage{algorithmic}
\usepackage[shortlabels]{enumitem}
\usepackage{url}

\usepackage{breakurl}

\usepackage{float}
\usepackage{graphicx}
\usepackage{dcolumn}
\usepackage{bm}
\usepackage{braket}
\usepackage{mathtools}
\usepackage{todonotes}
\usepackage{comment}
\usepackage{microtype}
\usepackage[numbers, sort&compress]{natbib}
\usepackage{textcase}
\usepackage{qcircuit}
\usepackage{hyperref} 

\newtheorem{proposition}{Proposition}
\usepackage{multirow}

\begin{document}
\title{Qubit-efficient quantum local search for combinatorial optimization}

\author{M.~Podobrii}
\email{mpd@terraquantum.swiss}
\affiliation{Terra Quantum AG, Kornhausstrasse 25, 9000 St. Gallen, Switzerland}
\author{V.~Kuzmin}
\affiliation{Terra Quantum AG, Kornhausstrasse 25, 9000 St. Gallen, Switzerland}
\author{V.~Voloshinov}
\affiliation{Terra Quantum AG, Kornhausstrasse 25, 9000 St. Gallen, Switzerland}
\author{M.~Veshchezerova}
\affiliation{Terra Quantum AG, Kornhausstrasse 25, 9000 St. Gallen, Switzerland}

\author{M. R. Perelshtein} 
\email{mpe@terraquantum.swiss}
\affiliation{Terra Quantum AG, Kornhausstrasse 25, 9000 St. Gallen, Switzerland}

\begin{abstract}
\noindent
An essential component of many sophisticated metaheuristics for solving combinatorial optimization problems is some variation of a local search routine that iteratively searches for a better solution within a chosen set of immediate neighbors. The size \(l\) of this set is limited due to the computational costs required to run the method on classical processing units. We present a qubit-efficient variational quantum algorithm that implements a quantum version of local search with only \(\lceil \log_2 l \rceil\) qubits and, therefore, can potentially work with classically intractable neighborhood sizes when realized on near-term quantum computers. Increasing the amount of quantum resources employed in the algorithm allows for a larger neighborhood size, improving the quality of obtained solutions. This trade-off is crucial for present and near-term quantum devices characterized by a limited number of logical qubits. Numerically simulating our algorithm, we successfully solved the largest graph coloring instance that was tackled by a quantum method. This achievement highlights the algorithm's potential for solving large-scale combinatorial optimization problems on near-term quantum devices.
\end{abstract}

\maketitle
\sloppy
\section{Introduction}\label{sec:Introduction}

\begin{figure}
    \centering
    \includegraphics[width=\linewidth]{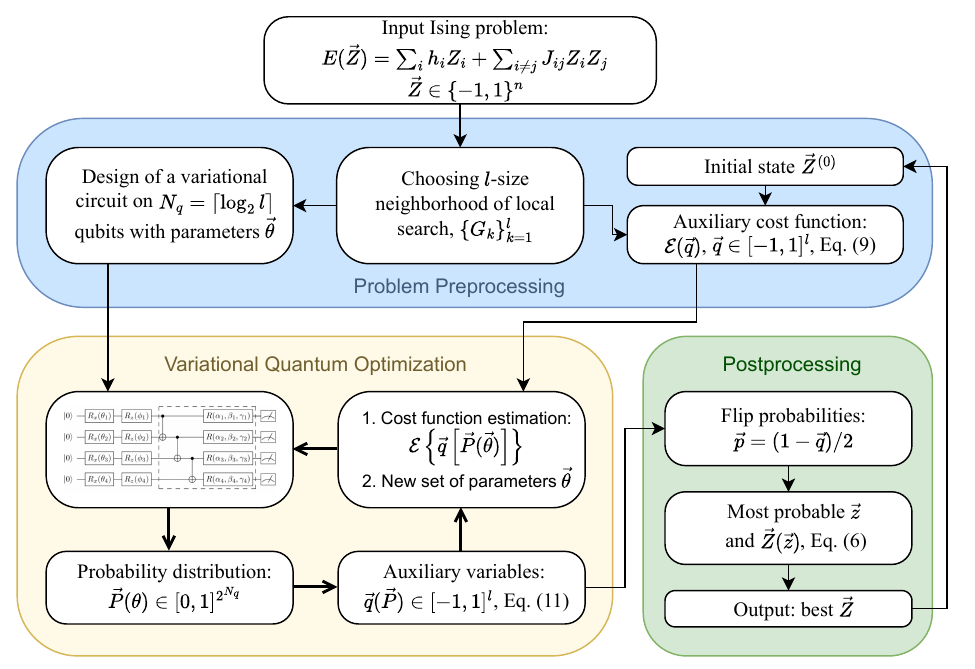}
    \caption{Scheme of our approach. We reduce the initial (discrete) Ising ground state problem to the optimization of a continuous \textit{auxiliary function} $\mathcal{E}(\vec{q})$ (Eq.~(\ref{eq:e_cont})). The auxiliary function is defined by choosing an initial state $\vec{Z}^{(0)}$ and groups $\{G_k\}$, where $G_k \subseteq [0{:}n{-}1]$. To find approximate local minima of the auxiliary function, we employ a variational quantum algorithm, enabling function evaluation even when $l$ is classically intractable. This is achieved by mapping the probability distribution $\vec{P}$ generated from quantum state measurements to the auxiliary variables $\vec{q}$ using Eq.~(\ref{eq:q_from_P}). After optimization of the quantum circuit, we estimate $\vec{p}=(1-\vec{q})/2$, which represents the probabilities of group flips. Using these probabilities, we sample several most probable $\vec{z}$, where $z_k=-1$ if the corresponding group $G_k$ is flipped and $1$ otherwise. We decode $\vec{z}$ into the original variables $\vec{Z}$ via Eq.~(\ref{eq:rev}) and choose the best $\vec{Z}$ with the minimal energy. Then, we restart the procedure from the beginning taking $\vec{Z}^{(0)}=\vec{Z}^{\mathrm{best}}$. We perform several such rounds until the solution is not improved.}
    \label{fig:scheme}
\end{figure}

Local search is one of the most fundamental heuristic algorithms successfully applied to many challenging combinatorial optimization problems. The so-called $r$-local search (or $r$-flip) algorithm~\cite{Alidaee2023} starts from a candidate solution and iteratively explores a chosen neighborhood, of size $l$, of solutions that differ by at most $r$ bits (or Hamming distance $r$). The process stops when all the neighbors are not better than the current solution. Many metaheuristics, such as simulated annealing~\cite{Kirkpatrick1983}, tabu search~\cite{Palubeckis2004_ts}, and memetic algorithms~\cite{Merz2005_memetic}, are based on local search~\cite{heuristics}. Increasing the size of the neighborhood by increasing $r$ can yield better solutions at the cost of an exponentially growing runtime that follows the binomial coefficient \(l \sim C_n^r \sim O(n^r)\) where $n$ is the total number of variables. Therefore, the parameter $r$ typically must be chosen as small.

In this work, we address the aforementioned problem by introducing a Variational Quantum Algorithm (VQA) that uses a parametrized quantum-circuit ansatz of sufficient depth to implement the $r$-local search heuristic. The distinct feature of our approach is that the number, $N_q$, of employed qubits adapts to the chosen parameter $r$, ranging from \(N_q=\lceil\log_2 n\rceil\) when $r=1$ to $N_q=n$ when $r=n$. In contrast to the previously proposed \textit{quantum local search} method \cite{Tomesh2022, Liu2023}, which reduces the number of qubits by breaking the problem into smaller subproblems, our approach addresses the entire problem directly and achieves qubit reduction through amplitude encoding. That is, for a chosen value of $r$, our method encodes the probabilities of flipping groups of bits, each involving no more than $r$ bits, into the square norm of the quantum state's amplitudes. Then, the parameters of the variational circuit are tuned so that the probabilities from the resulting quantum state minimize a designed non-linear objective function whose local minima coincide with the local minima of $r$-local search. We emphasize that, as demonstrated in our work, the required number of shots for estimating this cost function doesn't explicitly grow with $r$, allowing our quantum algorithm to handle a classically intractable number of neighbors. Finally, the optimized quantum state decodes into a set of multi-bit flips that are applied to the initial candidate to generate the resulting solution. The scheme of our algorithm is depicted in Fig.~\ref{fig:scheme}.

Previously, qubit-efficient approaches that require only \(N_q=\lceil\log_2 n\rceil\) qubits for solving $n$-variable optimization problems have been proposed in Refs.~\cite{Tan2021, Perelshtein2023nisqcompatible, tenecohen2023, PhysRevA.109.052441, PhysRevResearch.5.L012021}. However, such an efficient encoding automatically implies that the complexity of classically simulating the involved quantum circuit scales only polynomially with the problem size, limiting the opportunities for achieving quantum advantage on actual quantum hardware. Moreover, in our work, we show that the performance of the method proposed in Refs.~\cite{Tan2021, Perelshtein2023nisqcompatible}, referred to as \textit{minimal} encoding, is limited due to the presence of multiple local minima in the underlying cost function. While these functions can somehow be optimized with global optimization routines such as \textit{differential evolution} \cite{diff_evolution_scipy}, gradient-based methods will inevitably fall into local minima. In contrast, while increasing neighborhood size $l$ by employing a growing number, $N_q$, of qubits, our approach allows us to gradually eliminate these local minima. This leads to higher-quality solutions when the quantum variational ansatz is expressive enough.

At the opposite limit, $r=n$, our algorithm requires \(N_q=n\) qubits. In this case, all solutions are neighbors to each other, and all the local minima in the cost function are caused by the finite-depth circuit ansatz. This limit is similar to algorithms that use the so-called \textit{complete} encoding~\cite{Tan2021}, such as VQE \cite{Peruzzo2014, DezValle2021} and QAOA \cite{farhi2014quantum, Blekos_2024}, where the number of qubits is equal to the number of classical variables. In this work, we show that at this limit, our method gives compatible results to VQE with complete encoding, see Appendix~\ref{sec:Appendix/VQE_comparison}. 

Thus, our method intends to fill the gap between the \textit{minimal} and \textit{complete} encodings. In this regime, the involved quantum circuit might not be efficiently simulatable, potentially offering a quantum advantage. At the same time, for a given problem, our method allows the required number of qubits to be adjusted depending on the quantum device at hand. Some other approaches~\cite{Fuller2024, Patti2022} reduce the qubit number by only a constant factor. Reducing the number of qubits in a low-depth variational circuit could also improve the quality of solutions since the cost function of shallow circuits that have many qubits can suffer from severe under-parametrization, see Appendix~\ref{sec:Appendix/VQE_comparison}.

Other methods with an adaptive number of qubits were previously proposed in Refs.~\cite{Tan2021, sciorilli2024}. However, Ref.~\cite{Tan2021} shows the advantage only on sparse problems, while for dense problems, it demonstrates even worse results than minimal encoding due to the presence of contradictions in the procedure for decoding the solution from the quantum state. Furthermore, there is no mathematical justification for how increasing the qubit number affects the quality of the obtained solutions. In contrast, our method provides decoding without any contradictions. We show that it is able to reach the same solution quality as $r$-local search for $r>1$ and, therefore, that it outperforms \textit{minimal} encoding both on sparse and dense problems. 

While we were developing our approach, Sciorilli et al.~\cite{sciorilli2024} proposed an alternative quantum algorithm for solving combinatorial optimization problems using an adaptive number of qubits, demonstrating competitive results on the \textit{MaxCut} problems. Our method offers the possibility of constructing problem-specific objective functions. This advantage might be crucial for sparse and constrained problems. This is demonstrated in our numerical experiment on the graph coloring problem, where our approach successfully solves the largest instance tackled by a quantum algorithm to date.

The structure of this paper is organized as follows. In Sec.~\ref{sec:Objective}, we begin by introducing an auxiliary objective function whose local minima match those of local search over a specified neighborhood and discuss an efficient selection of the appropriate neighborhood for a given problem. Next, in Sec.~\ref{sec:Quantum}, we outline how to optimize the auxiliary function using a variational quantum algorithm. In Sec.~\ref{sec:Analysis}, we analyze the resource requirements to show that the quantum algorithm can handle a classically intractable number of neighbors. In Sec.~\ref{sec:Methods}, we specify the simulation details and the considered problems. Finally, in Sec.~\ref{sec:Results}, we present examples demonstrating our algorithm's competitive performance across various problems on a simulator and on a real IBM quantum device.

\section{Objective function}
\label{sec:Objective}
In our work, we focus on a class of quadratic unconstrained binary optimization (QUBO) problems, which is equivalent to the class of Ising models. These classes are general since most combinatorial optimization problems can be formulated as QUBO or Ising problems~\cite{qubo1, ising}.

QUBO problems are formulated as follows: for a given real upper-triangular matrix \(A\in\mathbb{R}^{n\times n}\), find a binary vector \(\vec{x}^*\in \{0,1\}^n\) that minimizes the function: 
\begin{equation}
C(\vec{x})=\vec{x}^TA\vec{x}=\sum_{i,j=1}^n a_{ij}x_ix_j=\sum_{i}^n a_{ii}x_i+\sum_{i<j}^n a_{ij}x_ix_j.
\label{eq:qubo}
\end{equation}

By transforming \(x_i \to (1-Z_i)/2\), QUBO problems can be reduced to the Ising model:
\begin{equation}
    E (\vec{Z})=\sum_i h_i Z_i +\sum_{i<j} J_{ij} Z_i Z_j,
\label{eq:ising}
\end{equation}
where $Z_i\in\{-1,1\}$ are spin variables. The coefficients are related as $J_{ij}=a_{ij}/4$, ${h_i=-\sum_k (a_{ik}+a_{ki})/4.}$

The original problem, given by Eq.~(\ref{eq:qubo}) or (\ref{eq:ising}), is defined on a discrete space. In the following, we define an auxiliary problem over continuous variables that can be optimized with a variational quantum algorithm and whose solution is equivalent to the original problem.

\subsection{Bilinear relaxation}
\label{sec:Objective/Bilinear}
The bilinear continuous relaxation of the problems (\ref{eq:qubo}) and (\ref{eq:ising}) are:
\begin{equation}
\text{QUBO:}~ \overline{C}(\vec{p})=\sum_{i} a_{ii}p_i+\sum_{i< j} a_{ij}p_ip_j,\quad p_i\in[0,1],
\label{eq:qubo_min}
\end{equation}
\begin{equation}
    \text{Ising:}~\overline{E}(\vec{q})=\sum_i h_i q_i +\sum_{i< j} J_{ij} q_i q_j,\quad q_i\in[-1,1],
\label{eq:ising_min}
\end{equation}
Assuming that \(x_i\) are independent random variables, the continuous function (\ref{eq:qubo_min}) is precisely~\cite{Boros2002PseudoBooleanO} the average value of the original discrete objective function (\ref{eq:qubo}) over a variable assignment distribution, where \(p_i=P(x_i=1)\) is the probability that \(x_i\) has value $1$. For the Ising model, consequently, \((1-q_i)/2\) is the probability \(P(Z_i=-1)\).

The functions (\ref{eq:qubo_min}) and (\ref{eq:ising_min}) have the same minimum as (\ref{eq:qubo}) and (\ref{eq:ising}) over discrete variables~\cite{Boros2002PseudoBooleanO}; therefore, optimization of the bilinear relaxation over the corresponding cube ($[0,1]^n$ for QUBO or $[-1,1]^n$ for Ising) is equivalent to the original discrete problem. 

We analyze the local minima to assess the solution quality achieved by optimizing the bilinear relaxation using gradient-based methods. The following propositions characterize the local minima of multilinear functions on the hypercube:
\begin{proposition}
    \label{th:loc_min}
    Let \(f(\vec{p})\) be a multilinear function. Then, on the hypercube \([a,b]^n\), all local minima are at its vertices - points in $\{a, b\}^n$. 
        
\end{proposition}
The proof of this proposition is given in Ref.~\cite{Laneve2010}. As an example, it means that for function (\ref{eq:qubo_min}), the local minima reside in \(\{0,1\}^n\).
\begin{proposition}
    \label{th:loc_min2}
    Let \(\overline{a}:=b,~\overline{b}:=a\) (e.g. for $x\in \{0, 1\}$ the $\overline{x}$ is \textit{logical negation}). $\vec{p}\in \{a,b\}^n$ is a local minimum of $f$ over $[a, b]^n$ $\iff f(p_1,\ldots,\overline{p_i},\ldots,p_n)>f(p_1,\ldots,p_i,\ldots,p_n)$ for all $i\in[1{:}n]$
\end{proposition}
\begin{proof}
Consider the vertex $A=(a,\ldots,a)$ without loss of generality. $A$ is a local minimum if and only if $\partial f /\partial p_i>0$ for all $i$. Since a multilinear function is linear by each variable when others are fixed:
\begin{equation}
    \frac{\partial f}{\partial p_i}(A)=\frac{f(a\ldots ,p_i=b, \ldots a)-f(a \ldots ,p_i=a, \ldots a)}{b-a}.
\end{equation}
Given that $b > a$, we obtain the desired statement.
\end{proof}
Applying the proposition~\ref{th:loc_min2} to function (\ref{eq:qubo_min}), we obtain that a solution is a local minimum of the bilinear relaxation if and only if it is better than all neighbor solutions differing by one bit. Note that this condition coincides with the local minimum of 1-local search. A gradient-based optimizer may get stuck in any local minimum; therefore, optimization of function (\ref{eq:qubo_min}) gives solutions of similar quality as 1-local search. As a consequence, the results can be very poor: for some constrained optimization problems, such as the traveling salesman problem or graph coloring, when the constraints are treated as penalties, every feasible solution is a local optimum for 1-local search.

\subsection{Encoding larger neighborhoods}\label{sec:Objective/EncodingNeighbors}
In bilinear relaxation, each variable encodes a single bit of the solution. As a result, the set of neighbors in the corresponding local search consists of only solutions that differ by one bit. To include more neighbors in the set, we add variables that control flipping groups of several bits. 

Starting from here, we focus on the Ising representation of the QUBO problem, for which the calculations are more straightforward. First, we choose the initial state of the spins to be $\vec{Z}^{(0)}$. Then, we introduce a new set of $l$ variables, \(z_k\in \{-1,1\}\), where each variable corresponds to a flip of a specific subset \(G_k\) of spins. States of the original Ising variables are decoded back by the following product:
\begin{equation}
    Z_i=Z^{(0)}_i\prod_{k:i\in G_k} z_k,
    \label{eq:rev}
\end{equation}
which counts how many times $i$-th spin is flipped within the groups. This definition implies that flipping \(z_k\) results in flipping all spins $Z_i$ in \(G_k\).
 
Substituting (\ref{eq:rev}) into (\ref{eq:ising}), we obtain:
\begin{equation}
    E(\vec{z})=\sum_i h_i Z^{(0)}_i\prod_{k:i\in G_k} z_k + \sum_{i < j} J_{ij} Z^{(0)}_i Z^{(0)}_j\prod_{k:i\in G_k} z_k \prod_{k:j\in G_k} z_k.
\label{eq:e_tmp}
\end{equation}
Since \(z_k^2=1\), Eq.~(\ref{eq:e_tmp}) can be rewritten in a multilinear form:
\begin{equation}
     E(\vec{z})=\sum_i h_i Z^{(0)}_i\prod_{k:i\in G_k} z_k + \sum_{i < j} J_{ij}Z^{(0)}_i Z^{(0)}_j \prod_{\substack{k:i\in G_k\\j\notin G_k}}z_k
         \prod_{\substack{k:i\notin G_k\\j\in G_k}}z_k. 
    \label{eq:e_discrete}
\end{equation}
As with bilinear relaxation, this allows the discrete variables \(z_k\in\{-1,1\}\) to be replaced with continuous variables \(q_k\in [-1,1]\). We call the resulting function the \textit{auxiliary} function:
\begin{equation}
     \mathcal{E}(\vec{q})=\sum_i h_i Z^{(0)}_i\prod_{k:i\in G_k} q_k + \sum_{i < j} J_{ij}Z^{(0)}_i Z^{(0)}_j \prod_{\substack{k:i\in G_k\\j\notin G_k}}q_k
         \prod_{\substack{k:i\notin G_k\\j\in G_k}}q_k. 
    \label{eq:e_cont}
\end{equation}
In a similar fashion to the aforementioned bilinear relaxation, Eq.~(\ref{eq:e_cont}) represents the average value of the Eq.~(\ref{eq:e_discrete}) over a variable assignment distribution where $z_k$ are independent and \((1-q_k)/2\) is the probability of \(z_k=-1\).

Applying proposition~\ref{th:loc_min2} to function~(\ref{eq:e_cont}), one can prove that the local minima of this function correspond to the solutions that are better than all neighboring solutions differing by spin-flips encoded in groups \(\{G_k\}\), i.e., we built an analog of a local search over a chosen neighborhood. When groups $G_k=\{k\}$ with $k\in [1{:}n]$ include only one bit and all $Z^{(0)}_i=1$, the auxiliary function (\ref{eq:e_cont}) matches with the bilinear relaxation (\ref{eq:ising_min}). If all subsets of spins are encoded, i.e., $\{G_k\}=\mathcal{P}(\{1,2,\ldots ,n\})$, where $\mathcal{P}(S)$ denotes a power set of $S$, then the local minima of the auxiliary function map to the global optima of the initial problem. However, in that case, the number of groups and, consequently, the number of variables in the auxiliary function equals $2^n$.

Note that Eq.~(\ref{eq:e_cont}) is a high-order polynomial. For instance, for $r$-flip groups, the degree of the product $\prod_{k:i\in G_k} q_k$ equals $O(n^{r-1})$. Since $q_k\in [-1,1]$, this leads to a barren plateau-like problem~\cite{McClean2018}, where function~(\ref{eq:e_cont}) has values close to zero for most of the landscape. Moreover, since taking the derivative reduces the degree only by 1, this issue also concerns the gradient, hessian, and so on, leading to an adverse flat landscape if $r$ is large enough. We mitigate this problem by an additional reparameterization described in Sec.~\ref{sec:Quantum/P_to_q}.

The presented approach can be straightforwardly extended to general PUBO problems, as shown in Appendix~\ref{sec:Appendix/HighOrder}.

\subsection{Neighbors choice}
\label{sec:Objective/NeighborsChoice}

For an optimization problem with all-to-all interactions, there is no general rule for choosing a set of neighbors for a given $r$-local search; thus, in this scenario, all solutions at Hamming distance at most $r$ can be included in the set of neighbors defined by \(\{G_k\}\).
However, for sparse problems, it makes sense to select only a problem-specific subset of the neighbors.

Indeed, let's consider a 1-local search or bilinear relaxation. The increment of the energy (\ref{eq:ising}) after inverting the \(k\)-th spin is
\begin{equation}
    \Delta_k =-2Z_k(h_{k}+\sum_{i\neq k}(J_{ki}+J_{ik})Z_i),
\end{equation}
where $Z_k$ is the spin value in the current solution.

From this follows that \(\Delta_k\) depends on \(Z_k\) and adjacent spins, i.e., those spins \(Z_i\) for which \(J_{ki}+J_{ik}\ne 0\). Therefore, the increment of the energy after the simultaneous inversion of two nonadjacent spins \((m,k)\) is \(\Delta_{mk}=\Delta_m + \Delta_k\). For the local minimum of 1-local search or the bilinear relaxation, \(\Delta_m>0\) and \(\Delta_k>0\), and therefore \(\Delta_{mk}>0\). In other words, if neither flipping the $m$-th spin nor flipping the $k$-th spin improves the solution, neither will the simultaneous flipping of both of them. That is, if we pursue the elimination of local minima, it makes sense to include only adjacent spin pairs into groups \(\{G_k\}\).

The same idea can be generalized to large spin sets: if a group of spins \(G\) is a union of disjoint groups \(G_1\), \(G_2\), where no elements \(k\in G_1\) and \(m\in G_2\) interact (\(J_{km}{=}J_{mk}{=}0\)), then the increment in the energy after inverting all the spins from \(G\) is \(\Delta_{G}=\Delta_{G_1}+\Delta_{G_2}\). If \(G_1\) and \(G_2\) are encoded, encoding \(G\) does not contribute to the elimination of the local minima.
Therefore, in general, it makes sense to encode only connected induced subgraphs of size $\leq r$ to implement $r$-local search.

Moreover, for constrained problems, additional groups can be excluded if we can ensure that flipping them in a feasible solution would lead to an infeasible one. An example of this is the graph coloring problem, which we consider in Sec.~\ref{sec:Results/Graph_coloring}.

\section{Variational quantum algorithm}
\label{sec:Quantum}
The main result of the previous sections is that we can imitate $r$-local search for an Ising problem by minimizing the auxiliary function given by Eq.~(\ref{eq:e_cont}). However, the number of variables (groups) \(l\) grows quickly with $r$: it varies from \(l=n\) for bilinear relaxation to \(l=2^n\) for the case where \(\{G_k\}\) are all subsets of spins, in which case all solutions are neighbors of each other. Thus, \(r\) is limited by computational capabilities. 

In this section, we address this problem by introducing a variational quantum algorithm aimed at finding a minimum of Eq.~(\ref{eq:e_cont}). Our approach uses a parametrized quantum circuit to prepare a quantum state and a specially designed algorithm to associate each obtained measurement outcome to some group \(G_k\). The probability of each outcome is converted to the probability of flipping the corresponding group. These flip probabilities can be used for estimating the auxiliary cost function in Eq.~(\ref{eq:e_cont}), which can then be minimized by updating the circuit parameters. When the circuit optimization is finished, the flip probabilities are used to sample the most probable solutions.

We acknowledge that the previous analysis of the local minima is no longer applicable due to the fact that the optimization targets the parameters of the quantum circuit rather than \(\vec{q}\) directly. Nevertheless, our numerical experiments for small-scale problems (see Sec.~\ref{subsec:large_maxcut}) demonstrate that if the circuit depth is sufficient, the solutions are similar to those obtained by classical local search over the corresponding neighborhood. Remarkably, we achieve this performance with the number of optimized circuit parameters much fewer than $l$.

In Sec.~\ref{sec:Quantum/k_to_G}, we show how to get the correspondence between the measurement outcomes and the groups. In Sec.~\ref{sec:Quantum/P_to_q}, we introduce the mapping between the outcome probabilities and the flip probabilities. In Sec.~\ref{sec:Quantum/Ansatz}, we specify the variational ansatz used in our algorithm. Finally, in Sec.~\ref{sec:Quantum/Optimization}, we describe the entire workflow.

\subsection{Mapping measurement output to group}
\label{sec:Quantum/k_to_G}
Hereinafter, instead of $\{G_k\}_{k=0}^{l-1}$, we define the groups as $\{G_{\mu}\}_{\mu=0}^{2^{N_q}-1}$, where $G_{\mu}$ is the group corresponding to measurement outcome $\mu$. Below, we show how to decode $\mu$ into $G_{\mu}$.

\paragraph{Small $l$ case:}
If the number of groups $l$ is small enough and they are explicitly stored in a list $\{G_k\}_{k=0}^{l-1}$, then we take \(\lceil \log_2 l \rceil\) qubits and assume that $G_{\mu}$ is simply the ${\mu}$-th group in this list, with measurement outcomes where ${\mu} \geq l$ being discarded. We refer to this straightforward mapping as \textit{trivial}. 

When the number of groups becomes classically intractable, they can be defined implicitly, e.g., ``all subsets of $[0{:}n{-}1]$ of size $\leq r$'' or ``all connected $\leq r$-size subgraphs of a graph induced by the original problem.'' For these cases, we suggest efficient methods for decoding $G_{\mu}$ from ${\mu}$. 

\paragraph{Full $r$-neighborhood:} 
Let's consider the case when the groups are all subsets of $[0{:}n{-}1]$ of size $\leq r$ (the neighbors are all solutions at Hamming distance at most $r$). The number of groups in this case is $l=\sum_{m=1}^r \binom{n}{m}$, where $\binom{n}{m}$ is a binomial coefficient.

One method for decoding is to convert the outcome number \(\mu\) into a number in base \(n\). The unique digits of this base-\(n\) representation can then be interpreted as elements of the group \(G_{\mu}\). For example, consider $n=16$ and projective measurements of the 12-qubit states $\ket{010010110011}$ and $\ket{100010001000}$. Since $010010110011_2=4B3_{16}$ and $100010001000_2=888_{16}$, the first state corresponds to the group $G_\mu=\{3,4,11\}$, where the elements order does not play a role, whereas the second one to a single-element group $G_\mu=\{8\}$.

For \(r\)-flip groups, we require numbers with up to \(r\) digits, which implies \(\mu \in [0{:}n^r{-}1]\). Note that $n^r>l$ since the digits may repeat, and thus, several $\mu$ map to one group. The number of qubits required for this method is then \(N_q = \lceil r \log_2 n \rceil\), which is larger than \(\lceil \log_2 l \rceil\) for the trivial mapping. However, the overhead is relatively modest: for the extreme case of $r=n$, the required number of qubits becomes \(\lceil n \log_2 n \rceil\) compared to $n$ for the trivial mapping.

The complexity of the base conversion does not exceed \( O(\log (r \log n) M(r\log n))\), where \(M(b) \leq O(b^2)\) is the complexity of \(b\)-bit number multiplication \cite{Brent2010}.

Alternatively, the full \(r\)-neighborhood allows for a \textit{trivial} mapping even when \(l\) is classically intractable, as detailed in Appendix~\ref{sec:Appendix/Mapping}. While this method does not have an overhead in the number of qubits, it comes with a higher time complexity compared to the method described previously.

\paragraph{Sparse neighborhood:} Now, consider a more sophisticated case where only connected subgraphs of the graph induced by the original problem need to be encoded, as described in Sec.~\ref{sec:Objective/NeighborsChoice}. Counting the number of connected subgraphs is a computationally challenging problem \cite{Welsh_1997, subgr1}; therefore, even determining \(l\) becomes infeasible for large \(r\). Nevertheless, the mapping remains possible. 

For simplicity, let us assume a \(d\)-regular graph. Any connected subgraph can be constructed by selecting some initial vertex \(s \in [0{:}n{-}1]\) serving as a reference and sequentially choosing neighbors of the previously selected vertex. In this representation, instead of defining the group as \(G=\{i_1, i_2, \ldots, i_r\}\), where \(i_p \in [0{:}n-1]\) indexes the vertices, the group can be represented as \(W=\{s, j_2, \ldots, j_r\}\), where \(j_p \in [0{:}d-1]\) indexes the neighbors at each step. 

The index \(\mu\) can be converted into \(W\) using a base conversion method similar to the one described earlier, with an adjustment to first identify the starting vertex \(s=\lfloor \mu / d^{r-1} \rfloor\) before applying the base-\(d\) conversion for subsequent elements. Here, \(\mu\in [0{:}nd^{r-1}{-}1]\) and the required number of qubits \(N_q = \lceil \log_2 n + (r-1) \log_2 d \rceil\).

\subsection{Mapping outcome probabilities to variables $\vec{q}$}
\label{sec:Quantum/P_to_q}
To map the probabilities $P_{\mu} \in [0, 1]$ of measurement outcome $\mu$ to the auxiliary function variables $q_{\mu}\in[-1, 1]$, we suggest using the following continuous monotonic transformation:
\begin{equation}
    q_{\mu}=2\frac{\tanh(\alpha(1-MP_{\mu}))+1}{\tanh \alpha +1}-1,
    \label{eq:q_from_P}
\end{equation}
given in Fig.~\ref{fig:q_plot}, where \(\alpha>0\) and \(M>0\) are additional hyperparameters. The parameter $\alpha$ makes the function more step-like, thereby reducing the adverse area where many $q_{\mu}\approx 0$, mentioned in Sec.~\ref{sec:Objective/EncodingNeighbors}. The parameter $M$ impacts how many $q_{\mu}$ can fall in a negative value: under transformation~(\ref{eq:q_from_P}), it happens only for values $P_{\mu} \gtrsim 1/M$. Since $P_{\mu} \in [0,1]$ and $\sum_{\mu} P_{\mu} \leq 1$, not more than $M$ values $q_{\mu}$ can be negative (see Tab.~\ref{tab:q_example} for an example). In Sec.~\ref{sec:Analysis/Shots}, we show that this limitation provides a good finite measurement estimate of the auxiliary function.

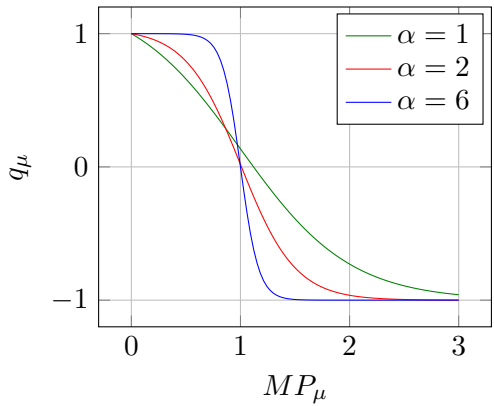
\begin{figure}[H]
\centering
\begin{tikzpicture}
\begin{axis}[grid,
	xlabel = {$MP_{\mu}$},
	ylabel = {$q_{\mu}$},
        width=0.45\textwidth
]
\addplot[domain=0:3, samples=100, color=green!50!black] {2*(tanh(1*(1-x))+1)/(tanh(1)+1)-1};
\addplot[domain=0:3, samples=100, color=red] {2*(tanh(2*(1-x))+1)/(tanh(2)+1)-1};
\addplot[domain=0:3, samples=100, color=blue] {2*(tanh(6*(1-x))+1)/(tanh(6)+1)-1};

\addlegendentry{\(\alpha=1\)}
\addlegendentry{\(\alpha=2\)}
\addlegendentry{\(\alpha=6\)}

\end{axis}
\end{tikzpicture}
\caption{Plots of the transformation $q_{\mu}(P_{\mu})$ for different values of the hyperparameter $\alpha$.}
\label{fig:q_plot}
\end{figure}

\begin{table}[ht!]
    \centering
\begin{tabular}{|c|c|c|c|c|c|c|c|c|c|}
    \hline
     &  & \multicolumn{8}{c|}{$\vec{P}$} \\
     \hline
    &  & $1/4$ & $1/4$ & $1/8$ & $1/8$ & $1/8$ & $1/16$ & $1/16$ & $0$ \\
    \hline
    \hline
     $M$ & $\alpha$ & \multicolumn{8}{c|}{$\vec{q}$} \\
     \hline
     $2$ & $1$ & $0.66$ & $0.66$ & $0.86$ & $0.86$ & $0.86$ & $0.93$ & $0.93$ & $1.00$ \\
     $4$ & $1$ & $0.14$ & $0.14$ & $0.66$ & $0.66$ & $0.66$ & $0.86$ & $0.86$ & $1.00$ \\
     $8$ & $1$ & \textcolor{red}{$-0.73$} & \textcolor{red}{$-0.73$} & $0.14$ & $0.14$ & $0.14$ & $0.66$ & $0.66$ & $1.00$ \\
     $16$ & $1$ & \textcolor{red}{$-0.99$} & \textcolor{red}{$-0.99$} & \textcolor{red}{$-0.73$} & \textcolor{red}{$-0.73$} & \textcolor{red}{$-0.73$} & $0.14$ & $0.14$ & $1.00$ \\
     \hline
     \hline
     $2$ & $2$ & $0.79$ & $0.79$ & $0.94$ & $0.94$ & $0.94$ & $0.98$ & $0.98$ & $1.00$ \\
     $4$ & $2$ & $0.02$ & $0.02$ & $0.79$ & $0.79$ & $0.79$ & $0.94$ & $0.94$ & $1.00$ \\
     $8$ & $2$ & \textcolor{red}{$-0.96$} & \textcolor{red}{$-0.96$} & $0.02$ & $0.02$ & $0.02$ & $0.79$ & $0.79$ & $1.00$ \\
     $16$ & $2$ & \textcolor{red}{$-1.00$} & \textcolor{red}{$-1.00$} & \textcolor{red}{$-0.96$} & \textcolor{red}{$-0.96$} & \textcolor{red}{$-0.96$} & $0.02$ & $0.02$ & $1.00$ \\
     \hline
     \hline
     $2$ & $3$ & $0.91$ & $0.91$ & $0.98$ & $0.98$ & $0.98$ & $0.99$ & $0.99$ & $1.00$ \\
     $4$ & $3$ & $0.00$ & $0.00$ & $0.91$ & $0.91$ & $0.91$ & $0.98$ & $0.98$ & $1.00$ \\
     $8$ & $3$ & \textcolor{red}{$-1.00$} & \textcolor{red}{$-1.00$} & $0.00$ & $0.00$ & $0.00$ & $0.91$ & $0.91$ & $1.00$ \\
     $16$ & $3$ & \textcolor{red}{$-1.00$} & \textcolor{red}{$-1.00$} & \textcolor{red}{$-1.00$} & \textcolor{red}{$-1.00$} & \textcolor{red}{$-1.00$} & $0.00$ & $0.00$ & $1.00$ \\ 
     \hline
\end{tabular}

    \caption{We consider a demonstrative probability distribution over a $3$-qubit state $\vec{P}=\{1/4,1/4,1/8,1/8,1/8,1/16,1/16,0\}$. For different values of $\alpha$ and $M$, we show the corresponding $\vec{q}$. As can be seen, the parameter $M$ affects the number of negative $q_{\mu}$: larger values of $M$ result in a larger number of negative $q_{\mu}$ for a given $\alpha$ and $\vec{P}$. When $\alpha$ is small, $q_{\mu}$ changes smoothly from $-1$ to $1$ as $P_{\mu}$ increases; if $\alpha$ is large, $q_{\mu}$ experiences an abrupt change when $P_{\mu}\approx 1/M$.}
    \label{tab:q_example}
\end{table}

The property $q_{\mu}(0)=1$ implies that if the outcome $\mu$ was not measured in $N$ rounds, the corresponding variable $q_{\mu}$ gets value $1$. Therefore, our encoding avoids the issue present in \textit{minimal} encoding \cite{Tan2021, Perelshtein2023nisqcompatible}, where outcomes that did not occur lead to an undefined bit value. 

\subsection{Variational circuit}
\label{sec:Quantum/Ansatz}
We use a common approach to address optimization problems that lack a problem-inspired circuit ansatz. We exploit a hardware-efficient ansatz~\cite{leone2024practical} and use gates that can be efficiently transpiled into the gates available on the IBM quantum device of type \textit{Eagle r3}\footnote{https://docs.quantum.ibm.com/guides/processor-types}. The structure of the ansatz on $N_q$ qubits, given in Fig.~\ref{fig:ansatz}, is chosen empirically. The ansatz consists of Hadamard gates on each qubit and a sequence of $L$ identical layers, consisting of parameterized \(Z\) and \(Y\) rotations and two-qubit entangling ECR gates, that have the following matrix representation:
\begin{equation}
    \mathrm{E C R} =\frac{1}{\sqrt{2}}\begin{pmatrix}
0 & 1 & 0 & i \\
1 & 0 & -i & 0 \\
0 & i & 0 & 1 \\
-i & 0 & 1 & 0
\end{pmatrix}
\end{equation}
The obtained circuit has $2N_q L$ trainable parameters $\vec{\theta}$.

\begin{figure}[H]
    \centering
    \scalebox{0.8}{
\Qcircuit @C=0.7em @R=0.2em @!R { \\
	 	\nghost{{q}_{0} :  } & \lstick{{q}_{0} :  } & \gate{\mathrm{H}} \barrier[0em]{5} & \qw & \gate{\mathrm{R_Z(\theta_{01})}}  & \qw & \multigate{1}{\mathrm{ECR}} & \qw  & \qw & \gate{\mathrm{R_Y(\theta_{07})}} \barrier[0em]{5} & \qw & \gate{\mathrm{R_Z(\theta_{13})}}  & \qw & \multigate{1}{\mathrm{ECR}} & \qw  & \qw & \gate{\mathrm{R_Y(\theta_{19})}} \barrier[0em]{5} & \qw & \meter\\
	 	\nghost{{q}_{1} :  } & \lstick{{q}_{1} :  } & \gate{\mathrm{H}} & \qw & \gate{\mathrm{R_Z(\theta_{02})}} & \qw & \ghost{\mathrm{ECR}} & \multigate{1}{\mathrm{ECR}} & \qw & \gate{\mathrm{R_Y(\theta_{08})} } & \qw & \gate{\mathrm{R_Z(\theta_{14})}} & \qw & \ghost{\mathrm{ECR}} & \multigate{1}{\mathrm{ECR}} & \qw & \gate{\mathrm{R_Y(\theta_{20})}} & \qw &  \meter\\
	 	\nghost{{q}_{2} :  } & \lstick{{q}_{2} :  } & \gate{\mathrm{H}} & \qw & \gate{\mathrm{R_Z(\theta_{03})}} & \qw & \multigate{1}{\mathrm{ECR}} & \ghost{\mathrm{ECR}} & \qw & \gate{\mathrm{R_Y(\theta_{09})}} & \qw & \gate{\mathrm{R_Z(\theta_{15})}} & \qw & \multigate{1}{\mathrm{ECR}} & \ghost{\mathrm{ECR}} & \qw & \gate{\mathrm{R_Y(\theta_{21})}} & \qw & \meter\\
	 	\nghost{{q}_{3} :  } & \lstick{{q}_{3} :  } & \gate{\mathrm{H}} & \qw & \gate{\mathrm{R_Z(\theta_{04})}} & \qw & \ghost{\mathrm{ECR}} & \multigate{1}{\mathrm{ECR}} & \qw & \gate{\mathrm{R_Y(\theta_{10})}} & \qw & \gate{\mathrm{R_Z(\theta_{16})}} & \qw & \ghost{\mathrm{ECR}} & \multigate{1}{\mathrm{ECR}} & \qw & \gate{\mathrm{R_Y(\theta_{22})}} & \qw &   \meter\\
	 	\nghost{{q}_{4} :  } & \lstick{{q}_{4} :  } & \gate{\mathrm{H}} & \qw & \gate{\mathrm{R_Z(\theta_{05})}} & \qw & \multigate{1}{\mathrm{ECR}} & \ghost{\mathrm{ECR}} & \qw & \gate{\mathrm{R_Y}(\theta_{11})} & \qw & \gate{\mathrm{R_Z(\theta_{17})}} & \qw & \multigate{1}{\mathrm{ECR}} & \ghost{\mathrm{ECR}} & \qw & \gate{\mathrm{R_Y(\theta_{23})}} & \qw &   \meter\\
	 	\nghost{{q}_{5} :  } & \lstick{{q}_{5} :  } & \gate{\mathrm{H}} & \qw & \gate{\mathrm{R_Z(\theta_{06})}} & \qw & \ghost{\mathrm{ECR}} & \qw & \qw & \gate{\mathrm{R_Y(\theta_{12})}} & \qw & \gate{\mathrm{R_Z(\theta_{18})}} & \qw & \ghost{\mathrm{ECR}} & \qw & \qw & \gate{\mathrm{R_Y(\theta_{24})}} & \qw &  \meter\\ }}
    \caption{Ansatz with 6 qubits and 2 layers}
    \label{fig:ansatz}
\end{figure}
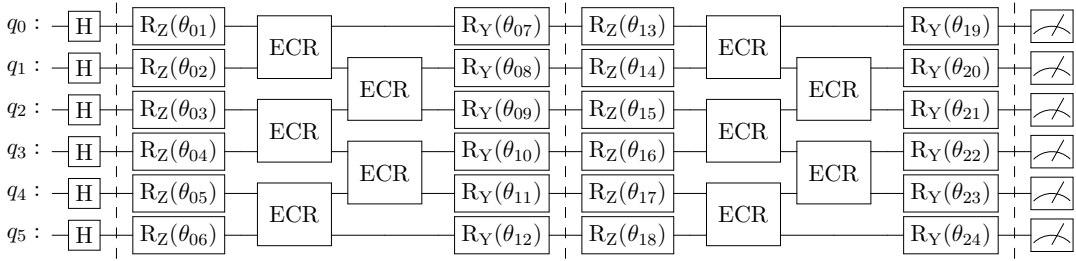

\subsection{Optimization process}
\label{sec:Quantum/Optimization}
After substituting Eq.~(\ref{eq:q_from_P}) into Eq.~(\ref{eq:e_cont}) and introducing a parameterized quantum circuit, we obtain a composite auxiliary function \(\mathcal{E}=\mathcal{E}\left[\vec{q}\left(\vec{P}(\vec{\theta})\right)\right]\). Here, we outline how to optimize it and get the solution to the original problem.
\paragraph{Initial solution selection:}
First, we choose an initial solution by fixing $\vec{Z}^{(0)}$ in Eq.~(\ref{eq:e_cont}). Typically, it may be a state with all spins set to one, a random state, or a random feasible solution for constrained problems.
\paragraph{Parameter Optimization:}
In order to estimate the auxiliary function, the quantum circuit is prepared and sampled \(N\) times in a computational basis resulting in the empirical probability distribution $\{N_{\mu}/N\}$ estimating $\vec{P}$, where \(N_{\mu}\) is the number of shots providing the \({\mu}\)-th outcome. Note that although vector $\vec{P}$ is $2^{N_q}$-dimensional, its estimate -- obtained with $N$ shots -- has no more than $N$ nonzero components, and this allows for efficient computations in a sparse format.

To optimize the auxiliary function using gradient-based methods, we also need to compute the gradient.
We are interested in 
\begin{equation}
    \frac{\partial \mathcal{E}}{\partial \vec{\theta}}=\frac{\partial \mathcal{E}}{\partial \vec{q}}\frac{\partial\vec{q}}{\partial \vec{P}}\frac{\partial\vec{P}}{\partial\vec{\theta}},
\end{equation}
where \({\partial \mathcal{E}}/{\partial \vec{q}}\) and \({\partial\vec{q}}/{\partial \vec{P}}\) are computed by direct differentiation of Eq.~(\ref{eq:e_cont}) and Eq.~(\ref{eq:q_from_P}).
We can compute \({\partial\vec{P}}/{\partial\vec{\theta}}\) using the \textit{parameter-shift rule} \cite{PhysRevA.99.032331}:
\begin{equation}
    \frac{\partial\vec{P}}{\partial\theta_j}=\frac{1}{2}\left(\vec{P}(\theta_j+\pi/2)-\vec{P}(\theta_j-\pi/2)\right).
    \label{eq:par_shift}
\end{equation}

The optimization process starts with a random initialization of the parameters $\vec{\theta}=\vec{\theta}_0$ from the uniform distribution on $[0, 2\pi)$. The parameters are then updated by some classical optimizer until convergence. After that, the circuit is sampled another $N$ times to obtain an estimate of $\vec{q}$.

\paragraph{Solution recovery:}
We consider every variable $z_k \in \{-1, 1\}$ of the auxiliary problem given by Eq.~(\ref{eq:e_tmp}) as an independent random variable distributed according to the Bernoulli distribution with "success probability" $p_k=P(z_k = -1) = (1 - q_k)/2$ as shown in Sec.~\ref{sec:Objective/EncodingNeighbors}. We sample a sequence of the $S$ most probable \(\vec{z}\) using the algorithm described in Appendix~\ref{sec:Appendix/Sampling} and select the one, denoted \(\vec{z}^{\mathrm{best}}\), with the minimal energy \(E(\vec{z})\) (Eq.~\ref{eq:e_tmp}). Then, we compute the resulting solution for the original problem \(\vec{Z}^{\mathrm{best}}\) using definition~(\ref{eq:rev}). We note that there are two different samplings in our algorithm: in the first one, the quantum circuit is sampled to obtain $\vec{p}$ -- the parameters of the Bernoulli distribution; in the second one, the solutions are sampled from this distribution.

\paragraph{Restarts:}
We remark that \(\vec{Z}^{\mathrm{best}}\) found in the previous paragraph is not necessarily a local minimum for the classical local search over a corresponding neighborhood. As previously mentioned, this is due to the fact that the quantum algorithm differs from the direct optimization of the auxiliary function. One limitation of the quantum algorithm arises from the $\vec{q}(\vec{P})$ mapping: as shown in Sec.~\ref{sec:Quantum/P_to_q}, the resulting $\vec{q}$ contains no more than $M$ negative components. Intuitively, that means the quantum algorithm can't provide more than $\approx M$ steps of classical local search started from $\vec{Z}^{(0)}$. That is, it makes sense to restart our algorithm taking $\vec{Z}^{(0)}=\vec{Z}^{\mathrm{best}}$. We perform $R$ such rounds, where $R$ is an additional hyperparameter.

\section{Resource analysis}\label{sec:Analysis}
In this section, we analyze the amount of resources, both quantum and classical, required for our algorithm.

\subsection{Number of layers}\label{sec:Analysis/Layers}
The quantum algorithm optimizes the circuit parameters $\vec{\theta}$ rather than \(\vec{q}\) directly. As a result, it performs optimization over a subspace of possible values of \(\vec{q}\), which may result in the appearance of new local minima \cite{larocca2024}. Moreover, in our approach, we encounter the traditional trade-off in variational quantum algorithms between the expressivity of the quantum circuit and the complexity of the parameter optimization. Indeed, an increased number of layers allows us to cover a larger subspace of the domain of variables $\vec{q}$ but increases the difficulty of parameter tuning.

In addition, our scheme gives freedom in choosing the required number of qubits to address the same optimization problem by selecting the number $l$ of groups to encode. Since fewer qubits generally require fewer layers to achieve good expressivity, we face another trade-off: we need to choose $l$ large enough to eliminate local minima of the bilinear relaxation while still allowing the auxiliary function to be optimized with a reasonable number of layers.
In Sec.~\ref{subsec:large_maxcut}, we provide numerical evidence that the number of layers doesn't need to grow linearly with $l$ to achieve good performance, showcasing the advantage of the quantum algorithm over classical optimization of the auxiliary function.

We acknowledge that our algorithm may encounter a barren plateau problem~\cite{McClean2018} inherent to the hardware-efficient circuit. Since our algorithm can utilize arbitrary parameterized circuits, it may be possible to choose a more sophisticated problem-inspired ansatz, which falls into the scope of further research.

\subsection{Number of measurements}\label{sec:Analysis/Shots}
As explained in Sec.~\ref{sec:Quantum/Optimization}, on a real quantum device, the auxiliary function is estimated from a finite number of circuit measurements. In this section, we analyze how many measurements (shots) are required to estimate the auxiliary function with a given accuracy. The inherent difficulty of the auxiliary function is that it is non-linear and can't be expressed as a quantum observable as in VQE and QAOA.

The estimate, $\hat{\mathcal{E}}$, for the true auxiliary function value, $\mathcal{E}$, is obtained by measuring the quantum circuit $N$ times and is equivalent to the value of the auxiliary function at a nearby point, $\hat{P}$. We treat $\hat{\mathcal{E}}$ as a random variable and evaluate the quality of the approximation using the mean squared error, $\text{MSE}=\mathbb{E}(\hat{\mathcal{E}}-\mathcal{E})^2$. The absolute value of the MSE depends both on the number of shots and on the flatness of the function landscape around the evaluation point.

To investigate the number of shots required to obtain a good estimate, we conduct an experiment using a 3-regular MaxCut-256 graph with weights randomly assigned in the interval $[-1,1]$ (see Sec.~\ref{sec:Methods}). For several values $r=1,3,5,6,7$, we take a randomly initialized circuit with $L=10$ layers and the number of qubits corresponding to the chosen $r$. We compute the exact value $\mathcal{E}$ of the auxiliary function by employing exact state-vector simulation and generate the set of 10000 estimates $\{\hat{\mathcal{E}}_i\}$, where each estimated value is obtained by sampling the quantum circuit $N$ times. Conducting the experiment for $N$ in the range $1$ to $2^{24}$, we compute the MSE of the estimates as a function of $N$. We repeat the experiment with a fixed $\alpha=3$ and various values of $M=2,20,200,2000$. The results are presented in Fig.~\ref{fig:shots}.

\begin{figure}
    \centering
    \includegraphics[width=\linewidth]{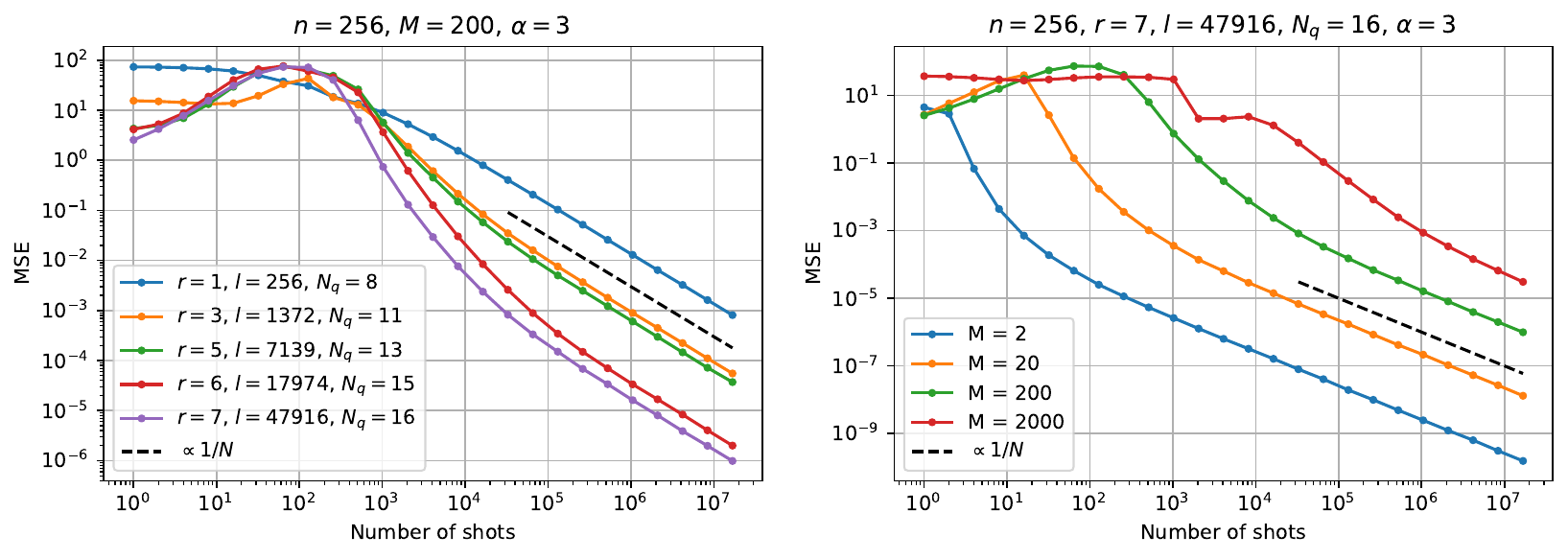}
    \caption{Mean squared error of the auxiliary function estimation obtained with a finite number of measurements (shots). Left: for different $r$; right: for different $M$.}
    \label{fig:shots}
\end{figure}

The right plot of Fig.~\ref{fig:shots} demonstrates that the MSE starts decreasing only when $N\gtrsim M$. Intuitively, such behavior can be explained as follows: if $N<M$, then the outcomes that are sampled at least once have an estimated probability $N_{\mu}/N>1/M$, which corresponds to $q_{\mu}<0$. In that case, the corresponding groups will be flipped even if the true probability $P_{\mu}$ is much smaller and doesn't correspond to a negative $q_{\mu}$. As a result, when $N<M$, the estimate becomes highly inaccurate, and only $N \gg M$ provides a good estimate. Furthermore, the plot indicates that when $N$ is sufficiently large, the MSE asymptotically scales as $\propto 1/N$. This behavior is similar to complete encoding (see Appendix~\ref{sec:Appendix/VQE_comparison}). 

The left plot of Figure~\ref{fig:shots} shows the impact of increasing $l$ with a fixed $M=200$. Remarkably, the behavior does not change significantly as $l$ grows: the MSE still starts decreasing only when $N\gtrsim M$ and reaches the same asymptotic when $N \gg M$. At first glance, one may notice the surprising trend that, for a fixed \( N \), the MSE decreases as \( l \) increases. This may be explained as follows: increasing $l$ worsens the flat landscape problem mentioned in Sec.~\ref{sec:Objective/EncodingNeighbors}. As a result, the variance of the auxiliary function across the landscape decreases, reducing the absolute value of the MSE. That, however, doesn't imply that fewer shots are needed for optimization with large $l$ since the accuracy requirements for optimization may increase accordingly.

In summary, the experiment demonstrates that despite the complex non-linear auxiliary objective function, the number of shots required to estimate it with a fixed accuracy doesn't grow with $l$. The only relevant parameter is $M$: when $N \gg M$, the estimation error reduces with $N$ in a similar manner to complete encoding. We emphasize that our analysis focuses on the number of shots required for fixed accuracy estimation rather than for optimization, which depends on the required accuracy. The latter issue is related to the barren plateau problem~\cite{McClean2018}, a fundamental issue affecting all variational quantum algorithms, and thus lies beyond the scope of this work.

\subsection{Complexity of the auxiliary function and gradient evaluation}
Each product in the auxiliary function given by Eq.~(\ref{eq:e_cont}) contains at most $l$ multipliers, and the entire expression includes \(O(n^2)\) summands. Therefore, the complexity of evaluating the function does not exceed \(O(n^2 l)\) in the general case. However, while optimizing with a finite number of measurements, \(N\), there are not more than \(N\) nonzero \(P_{\mu}\). Since \({q_{\mu}(P_{\mu}=0)=1}\) does not modify the auxiliary function (\ref{eq:e_cont}), it can be efficiently computed in at most \(O(n^2 N)\) classical computing time. 

For gradient estimation, the parameter shift rule requires the execution of two circuits for each parameter. 
Since $\vec{P}$ has no more than $N$ nonzero components, Eq.~(\ref{eq:par_shift}) implies that there are no more than $2N$ nonzero components of $\partial\vec{P}/{\partial\theta_j}$. Each nonzero component requires the evaluation of $\partial q_{\mu} / \partial P_{\mu}$, which can be done in $O(1)$, and $\partial \mathcal{E}/ \partial q_{\mu}$, which, similarly to $\mathcal{E}$, requires \(O(n^2 N)\) time. As a result, we get $O(n^2 N^2)$ for one parameter and $O(n^2 N^2 L N_q)$ for the entire gradient. Recall that for the trivial mapping, $N_q = O(\log l)$.

Given that $L$ and $N$ can be chosen much lower than $l$, as discussed in the previous sections, it follows from the analysis above that it's possible to estimate the auxiliary function and its gradient even if $l$ is classically intractable.
\section{Methods}
\label{sec:Methods}
\subsection{Simulation}
For all the numerical experiments described in Sec.~\ref{sec:Results}, we utilize a noiseless simulator from the \texttt{TensorCircuit} \cite{Zhang2023tensorcircuit} library with a \texttt{Jax} backend, which allows for automatic objective function differentiation. All the optimizations, except those in Sec.~\ref{sec:Results/QPU}, are performed with a \texttt{SciPy} implementation of the L-BFGS-B \cite{Byrd1995ALM} optimizer with the default settings.

Since we are working with heuristic algorithms that may yield different solutions on each run, we repeat each experiment multiple times to gather statistics. For each run, we choose a random initial solution $\vec{Z}^{(0)}$ and random circuit parameters $\vec{\theta}_0$. The number of groups in our experiments is relatively small, which allows us to store them explicitly and implement the \textit{trivial} mapping described in Sec.~\ref{sec:Quantum/k_to_G}.

\subsection{Local search}
\label{sec:Methods/Local_search}
For classical local search, we implement the first move improvement strategy \cite{heuristics}, where at each step, the first neighbor that is better than the current solution is accepted. The flipping groups that generate the neighbors are ordered by size, and within the same size, they are ordered lexicographically. For example, if $n=3$ and $r=2$, the groups are ordered as \(\{0\},\{1\},\{2\},\{0,1\},\{0,2\},\{1,2\}\).
\subsection{Problems}
\paragraph{MaxCut}
A canonical and well-studied combinatorial optimization problem is the MaxCut problem, where the aim is to find the partition of the graph's nodes into two complementary sets, such that the total weight of the edges between these two sets is as large as possible. MaxCut is equivalent to the Ising model (\ref{eq:ising}) with $h_i=0$ and $J_{ij}=w_{ij}$, where $w_{ij}$ is the weight of the edge between $i$th and $j$th nodes in the graph.

As a performance metric, we use the approximation ratio $\eta=E/E_\mathrm{opt}$, where $E$ is the energy of the found solution and $E_\mathrm{opt}$ is the energy of the optimal solution. To find the optimal solution, we use SCIP solver \cite{scip}.

\paragraph{Graph coloring}
Proper coloring in a graph $G = (V, E)$ is an assignment of colors to the vertices in a way that no two adjacent vertices are assigned the same color. For a given number of colors, $k$, we consider the optimization problem that aims to minimize the amount of improperly colored edges. By introducing binary variables $x_{v,i}$, where $v\in V, i \in \{1, \dots, k\}$ and $x_{v,i}$ is equal to one if vertex $v$ is assigned the color $i$, we can formulate the above problem as a QUBO problem~\cite{gc_qubo}:
\begin{equation}
    C(\vec{x}):= \lambda\sum_{v\in V} \left(1 - \sum_{i=1}^k x_{v,i} \right)^2 + \sum_{(v, w) \in E} \sum_{i=1}^k  x_{v,i}x_{w,i} \to \min\limits_{x_{v,i}{\in}\{0,1\}}.
    \label{eq:gr_col}
\end{equation}

It's possible to decode a binary assignment in a coloring only if each vertex is assigned to strictly one color, i.e., $\sum_{i=1}^k x_{v,i}=1$. The first term in Eq.~(\ref{eq:gr_col}) with prefactor \(\lambda >0\) aims to penalize the unfeasible solutions. If the graph can be colored in $k$ colors, the optimal assignment $\vec{x}$ corresponds to a proper coloring and has $C(\vec{x}) = 0$.

\section{Results}
\label{sec:Results}
This section presents the results obtained from the numerical simulation of our algorithm and its implementation on a QPU. We start with MaxCut problems. First, we employ algorithms based on bilinear relaxation (i.e., we set $r=1$ and $l=n$), and then we show how to improve the performance using extra qubits to encode more groups. After that, we consider the graph coloring problem as an example of a constrained problem where encoding a larger problem-specific neighborhood is essential to getting a good solution. Finally, we run an experiment on a real IBM quantum device.
\subsection{Bilinear relaxation}
\label{sec:Results/Minimal_encoding}
First, we investigate the performance of quantum algorithms based on bilinear relaxation. For this numerical experiment, we choose the MaxCut problem on two graph instances with $256$ nodes: one is a 3-regular graph, and the other is a fully connected graph. In both cases, the weights are sampled from a uniform distribution on $[-1,1]$.

We benchmark our method with $r=1$ against the \textit{minimal} encoding algorithm from Refs.~\cite{Tan2021, Perelshtein2023nisqcompatible}. Since $r=1$ implies $l=n=256$, our method requires 8 qubits, whereas the minimal encoding needs one additional ancilla qubit. 
For a fair comparison, in both cases, we optimize the same ansatz depicted in Fig.~\ref{fig:ansatz} with L-BFGS-B using exact circuit simulation. For our method, we set the hyperparameter values: $\alpha=2$, $M=256$, $S=1$, $R=1$. We optimize the circuit with the number of layers ranging from $2$ to $14$.

To provide additional insights, we run the classical 1-local search algorithm described in Sec.~\ref{sec:Methods/Local_search}. We also plot the results obtained with a direct optimization of the bilinear relaxation (Eq.~(\ref{eq:ising_min})) with the L-BFGS-B method. To gather statistics, we repeat all the experiments 50 times as explained in Sec.~\ref{sec:Methods}. The results are presented in Fig.~\ref{fig:min_enc}.
\begin{figure}
    \centering
    \includegraphics[width=\linewidth]{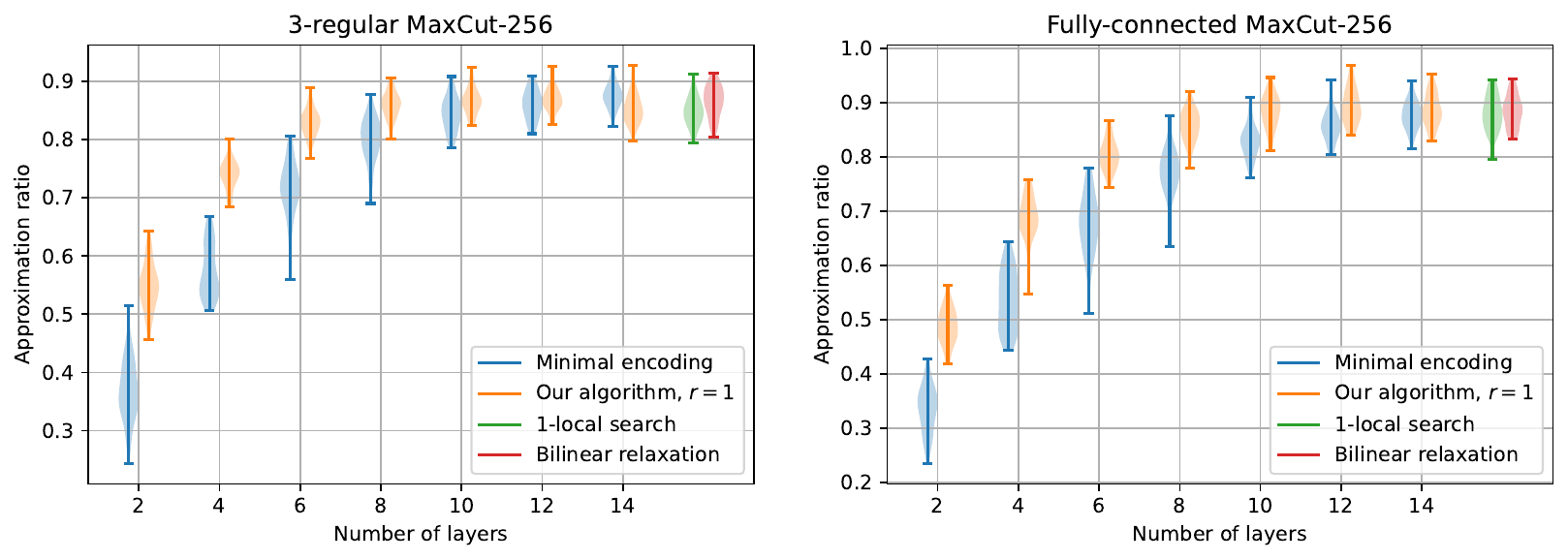}
    \caption{Approximation ratios of solutions obtained by our algorithm with $r=1$ in comparison to \textit{minimal} encoding \cite{Tan2021, Perelshtein2023nisqcompatible}, 1-local search and the direct optimization of the bilinear relaxation (\ref{eq:ising_min}) with the L-BFGS-B method.}
    \label{fig:min_enc}
\end{figure}

This numerical experiment confirms that classical optimization of the bilinear relaxation given by Eq.~(\ref{eq:ising_min}) provides solutions comparable to those obtained via 1-local search, with the performance of the quantum algorithms plateauing approximately on the same level. Additionally, when the number of layers is small, our method, even with $r=1$, outperforms the previously proposed quantum minimal encoding approach.

\subsection{Encoding larger neighborhood}
\label{subsec:large_maxcut}

We now demonstrate how the solutions can be improved by the use of extra qubits to encode a larger number of groups. For this experiment, we selected a 3-regular MaxCut problem with \(n=512\) nodes. The edge weights were randomly sampled from a uniform distribution over \([-1, 1]\). Our algorithm was tested with varying numbers of encoded groups. Since the problem is sparse, we take all connected groups \(G\) of size \(\leq r\) for four different values \(r=1,2,3,4\), as we explained in Sec.~\ref{sec:Objective/NeighborsChoice}.
The hyperparameters were fixed at \(M = n = 512\), \(\alpha = 7\), \(S=n=512\) and \(R=10\). The algorithm was launched from 48 random initial solutions \(\vec{Z}_0\). For comparison, a classical local search was run within the same neighborhood from the same initial solutions. 

To evaluate the results, the ensemble of final solutions from independent runs was compiled into an empirical cumulative distribution function (ECDF) of the objective function values. The ECDF for 16 layers is shown on the right side of Fig.~\ref{fig:maxcut512}. On the left side, we show the average approximation ratio over all runs for the given hyperparameters.

\begin{figure}
    \centering
    \includegraphics[width=\linewidth]{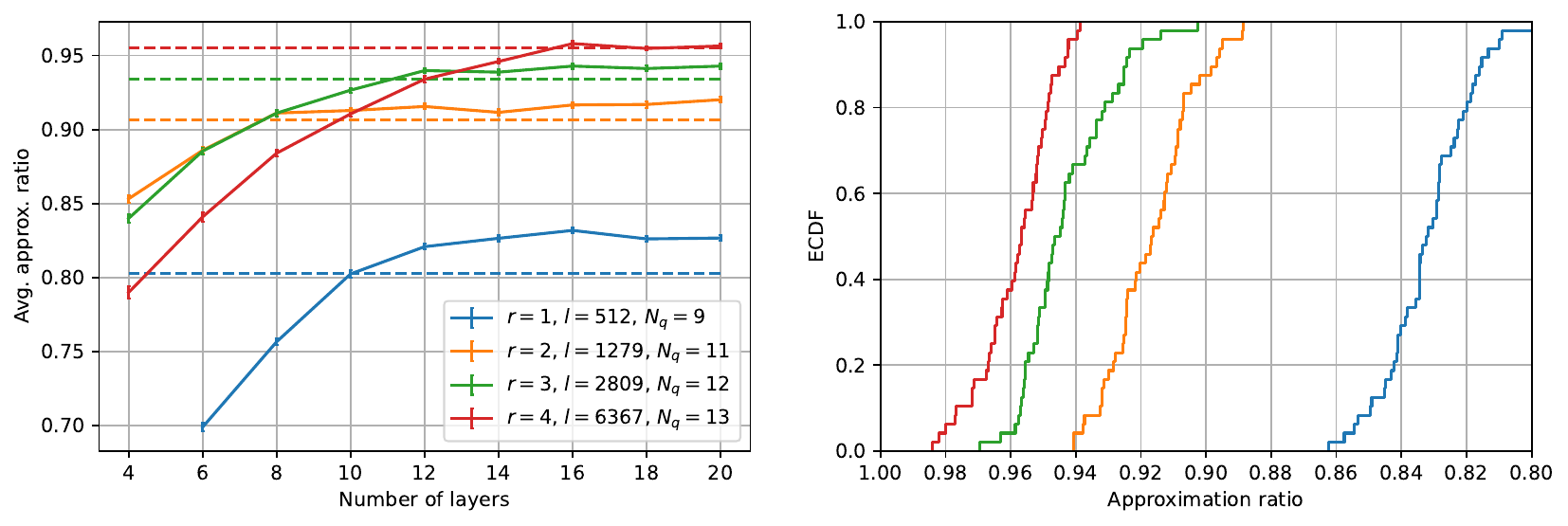}
    \caption{Left: the average approximation ratio for MaxCut-512 found by the algorithm with different hyperparameter configurations: number of neighbors and number of layers. The dashed lines indicate the average approximation ratio of the corresponding classical (discrete) local search. Right: the empirical cumulative distribution function (ECDF) of the approximation ratios of the solutions found by the quantum algorithm with 16 layers.}
    \label{fig:maxcut512}
\end{figure}
\begin{figure}
    \centering
    \includegraphics[width=0.7\linewidth]{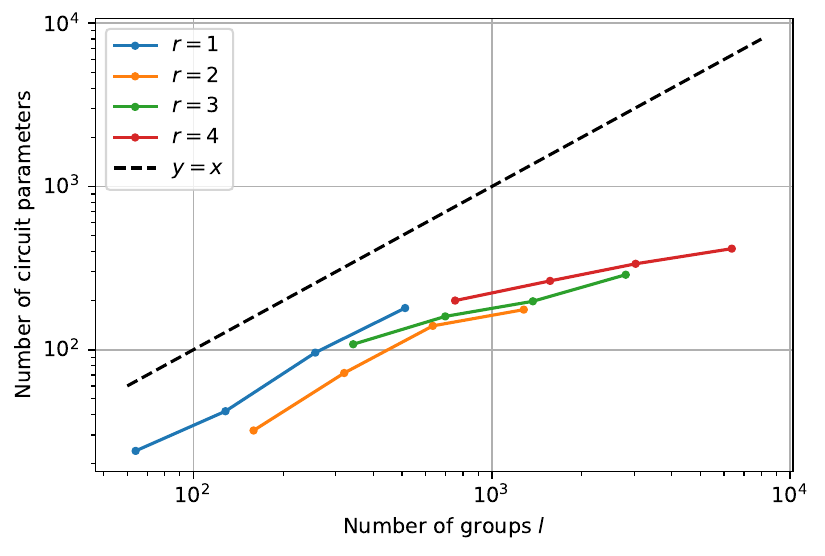}
    \caption{Number of circuit parameters required for reaching the same average approximation ratio as classical local search}
    \label{fig:scaling}
\end{figure}

The results indicate that, given a sufficiently deep circuit, increasing the number of groups improves the quality of the solutions. Notably, for small $r$, the quantum algorithm achieves solutions that are slightly superior to those obtained through classical local search within the corresponding neighborhood. We attribute this improvement to the multiple optimization rounds performed by the quantum algorithm.

To estimate the scaling of the required circuit depth, we repeat the experiment for a set of four 3-regular graphs with $n=64,128,256,512$ nodes. For each problem and each $r$, we define the number of layers required to reach the same average approximation ratio as the corresponding classical local search. We compare the number of circuit parameters with the number of neighbors $l$. The results are presented in Fig.~\ref{fig:scaling}. 

Remarkably, the required number of variational parameters is much smaller than $l$, and the advantage becomes more evident with the growth of $r$. This highlights a key advantage of the quantum algorithm over the classical optimization of the auxiliary function. While the classical approach requires optimizing \(l\) parameters corresponding to the total number of neighbors, the quantum algorithm achieves comparable or better performance with significantly fewer variational parameters.

\subsection{Graph coloring}\label{sec:Results/Graph_coloring}
Here, we present an example of a problem-specific approach to selecting groups for encoding, focusing on a constrained problem: the graph coloring problem. 

The minimal encoding approach is unsuitable for this problem because flipping a single bit in any feasible solution results in an infeasible one. Consequently, all feasible solutions become local minima due to the penalty term introduced in Eq.~(\ref{eq:gr_col}). 

To make a neighborhood of one feasible solution contain another feasible solution, we choose such variables in the auxiliary function such that each of them flips a pair of bits \((x_{v,i}, x_{v,j})\). The set of such variables for all \(v\) and \(i \ne j\) contains all groups that switch the vertex color. Note that such a choice is more efficient than taking all connected pairs since the latter contains pairs \((x_{v,i}, x_{w,i})\) that do not preserve feasibility.

We take the instance \textit{myciel7} from the benchmark in Ref.~\cite{gc_instances} for the experiment. This graph, $G = (V,E)$, has \(|V|=191\) vertices , $|E|=2360$ edges and \(k=8\) colors. Our experimental parameters are: $l=5348$ groups, $N_q=13$ qubits, $L=20$ layers, \(M=1000\), \(\alpha=4\), $S=10$, $R=4$. 

We repeat the experiment 100 times, each starting from a random feasible solution. In this setup, our algorithm successfully found the correct graph coloring in 19 out of 100 runs. To the best of our knowledge, this is the largest graph coloring instance solved with a quantum algorithm \cite{gc_qubo}. Note that algorithms based on the \textit{complete} encoding, including QAOA, would require 1528 (logical) qubits for this instance, which is far from the capabilities of current quantum devices.

\subsection{Experiments on a QPU}
\label{sec:Results/QPU}
Here, we test the applicability of the algorithm on a real noisy quantum device. To begin, we implement the variational optimization of the auxiliary cost function for solving a given MaxCut problem on a numerical simulator that emulates a quantum circuit with ideal gates subject to a finite number of measurements. Then, we test the obtained parameters on a real noisy IBM QPU and compare the outcome with the results from the numerical simulator. 

For the experiment, we chose a MaxCut problem on a 3-regular graph with 128 nodes, where the edge weights are randomly sampled from the interval \([-1,1]\). The algorithm is executed with hyperparameters \(r=1\), \(L=8\), \(M=128\), \(\alpha=2\) and \(R=1\). The circuit optimization is performed using the Simultaneous Perturbation Stochastic Approximation (SPSA) algorithm \cite{Spsa} on a numerical simulator with ideal gates and finite measurement sampling. Specifically, the probability distribution, \(\vec{P}\), is approximated by sampling \(N=1000\) shots from the exact state vector. The convergence plot of over 15,000 iterations, shown on the left of Fig.~\ref{fig:qpu}, illustrates the optimization process under these conditions.

Following this, the optimized circuit is executed on the simulator and the IBM \textit{brisbane} QPU \cite{ibm_quantum} with \(N=10000\) shots to estimate \(\vec{q}\) and extract the most probable solutions (see Sec.~\ref{sec:Quantum/Optimization}). The best solutions found for varying numbers of the most probable samples are depicted on the right of Fig.~\ref{fig:qpu}. As expected, the solutions sampled from the QPU are inferior to those from the simulator due to noise. However, our approach of sampling multiple probable solutions from \(\vec{q}\) drastically improves the quality of the solutions found with a noisy QPU: while QPU-induced noise distorts \(\vec{q}\), taking only the most probable solution would yield a significantly worse result compared to the simulator. By instead leveraging multiple high-probability solutions from the distribution, we partially mitigate the effects of the noise and achieve solutions that are closer to those generated by the simulator. 
\begin{figure}
    \centering
    \includegraphics[width=\linewidth]{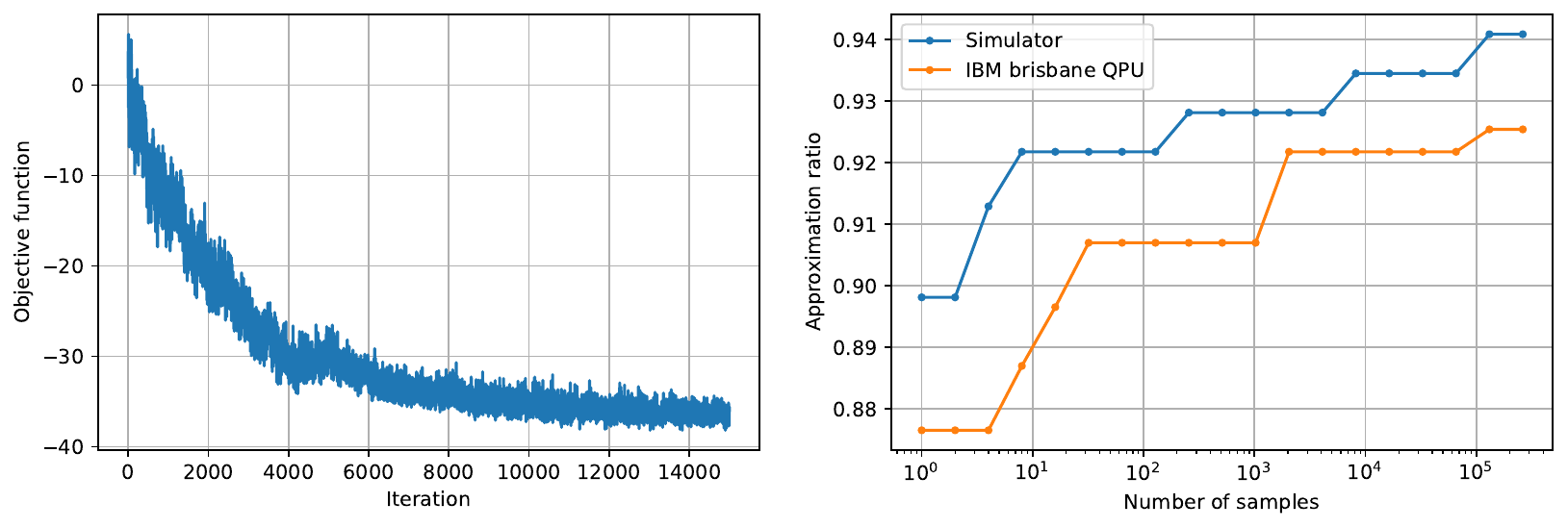}
    \caption{Left: Convergence of parameter optimization with SPSA over iterations using a numerical simulation that emulates the finite sampling of a quantum circuit with ideal gates. Right: approximation ratios of the best solution from the optimized circuit executed on the simulator and IBM \textit{brisbane} QPU, plotted against the number of most probable samples drawn from the distribution generated by the flip probabilities.}
    \label{fig:qpu}
\end{figure}

\section{Conclusion}\label{sec:Conclusion}
In this paper, we introduced a variational quantum algorithm tailored for solving combinatorial optimization problems, offering a flexible encoding scheme that can utilize a variable number of qubits, ranging from the logarithm of the number of classical variables up to the number of classical variables itself. This adaptability bridges the gap between minimal and complete encodings, addressing the trade-off in near-term quantum computing, where the minimal encoding often yields suboptimal solutions, while the complete encoding demands an impractical number of qubits. Through analytical insights and numerical experiments, we demonstrated that our algorithm achieves comparable performance to classical local search within a defined neighborhood while potentially offering an advantage to explore a classically intractable set of neighbors.

We implement our algorithm on various problems where encoding a large neighborhood is essential for achieving competitive solutions. Our approach successfully tackles problems that are currently out of reach for other quantum algorithms like minimal encoding, QAOA, and VQE. Despite these advancements, the algorithm's reliance on a general hardware-efficient variational ansatz poses challenges, including large circuit depths and potential trainability issues. These limitations highlight the need to develop problem-specific ansätze to enhance both performance and scalability. Addressing these challenges constitutes the scope for future work. Overall, our research offers a promising pathway for advancing quantum optimization algorithms on near-term devices, with the potential to solve more complex, real-world problems as quantum hardware continues to evolve.

\begin{acknowledgments}
We thank Katsiaryna Tsarova and Dr. Chris Mansell for their valuable discussions and suggestions.
\end{acknowledgments}

\newpage

\appendix

\section{Comparison with complete encoding}
\label{sec:Appendix/VQE_comparison}
\subsection{Complete encoding review}
Let us briefly review the complete encoding in the context of combinatorial optimization. This encoding requires a number of qubits equal to the number of classical variables: $N_q=n$. The objective is to minimize the average energy:
\begin{equation}
    \braket{E}=\braket{\psi| H |\psi},
    \label{eq:vqe_energy}
\end{equation}
where $H$ is the Ising Hamiltonian. Since it is diagonal in the computational basis, Eq.~(\ref{eq:vqe_energy}) can be rewritten as the average energy over a probability distribution over all possible states. It's convenient to use QUBO notation here since the bitstring, $\vec{\mu}$, which denotes the quantum state matches the corresponding solution $\vec{x}$.
\begin{equation}
    \braket{C}=\sum_{\mu} P_{\mu} C(\vec{\mu}),
    \label{eq:vqe_energy_mod}
\end{equation}
where $C(\vec{\mu})$ is the QUBO objective function (Eq.~\ref{eq:qubo}) and $P_{\mu}=|\braket{\mu|\psi}|^2$.
The estimate of Eq.~(\ref{eq:vqe_energy_mod}), obtained by sampling the quantum circuit $N$ times, is 
\begin{equation}
    \braket{\hat{C}}=\frac{1}{N} \sum_{j=1}^N C(\vec{\mu}^{(j)}),
    \label{eq:vqe_energy_est}
\end{equation}
where $\vec{\mu}^{(j)}$ is the $j$-th measurement outcome. The estimate is unbiased, and its variance decreases $\propto 1/N$~\cite{Guerreschi2017}.
The MSE is equal to the variance since the estimate is unbiased. Note that it doesn't explicitly depend on the number of qubits.
\subsection{Comparison with our encoding}
To see the similarity with our method, consider any solution, $\vec{x}$, as a result of flipping the group $G=\{i:x_i=1\}$ in an initial zero bitstring. For complete encoding, the value of $P_{\mu}$ then represents ``flipping probability,'' where the flip of only one group is allowed. In our auxiliary function, each group is independently flipped with probability $p_{\mu}=(1-q_{\mu})/2$. We can make our encoding quite similar to the complete one as follows:
\begin{enumerate}
    \item Encode all the subsets of the spins (including the empty set). Then, the total number of groups is $2^n$.
    \item Use the mapping $G_{\mu}=\{i:{\mu}_i=1\}$. In this case, there is a one-to-one correspondence between a qubit and a spin, just as with complete encoding.
    \item Choose $\vec{x}^{(0)}=\vec{0}$ ($\vec{Z}^{(0)}=\vec{1}$) as inital solution. In that case, flipping $G_{\mu}=\{i:{\mu}_i=1\}$ in $\vec{x}^{(0)}$ produces the solution $\vec{\mu}$. Now the outcome $\mu$ is related to the solution $\vec{\mu}$, as in complete encoding.
    \item Choose a small $1\lesssim M\lesssim 2$. We thereby allow no more than one negative $q_{\mu}$, ensuring that the most probable solution differs by at most one group flip from the initial $\vec{x}^{(0)}=\vec{0}$. We emphasize that the auxiliary function is still the average energy of the probability distribution over all (degenerate) $2^l=2^{2^n}$ states $\vec{z}$, where each group is either flipped or not. However, the small $M$ provides a low probability of multiple flips. In addition, as shown in Fig.~\ref{fig:shots} in the main text for $M=2$, the estimation error decreases starting from a few shots, which makes the estimation no harder than for complete encoding.
\end{enumerate}

\subsection{Numerical experiment}
We benchmark our algorithm against VQE by taking a fully connected random Ising instance of size $n=20$. We configure our algorithm according to the following scheme: we encode all subsets ($r=20$), set $\vec{Z}^{(0)}=\vec{1}$, $M=2$ and $\alpha=2$. We take 100 random initial solutions for each number of layers and perform $R=1$ round of optimization for each of them on a noiseless simulator. For comparison, we run VQE with the same ansatz and our algorithm with $r=1$, $M=20$, $\alpha=2$ and $r=2$, $M=40$, $\alpha=3$. For VQE, after optimization, we take 500 samples and choose the best one. For our algorithm, we perform one round ($R=1$), take $N=500$ quantum circuit samples and the $S=500$ most probable solutions. Such a choice is fair in that for both VQE and our algorithm, the circuit is sampled 500 times and the objective function is evaluated for 500 solutions. For each number of layers, we plot the average approximation ratio of the solutions obtained from the ensemble of 100 initial points. For the best number of layers, we plot the entire empirical cumulative distribution function. The results are presented in Fig.~\ref{fig:small}.
\begin{figure}
    \centering
    \includegraphics[width=\linewidth]{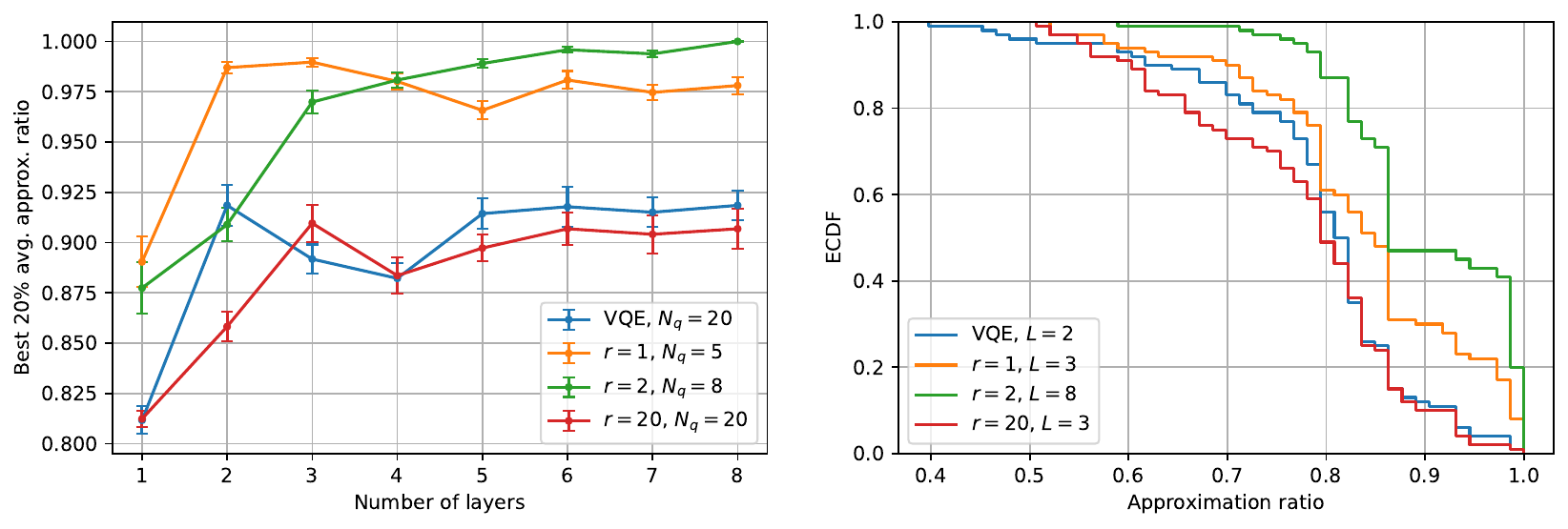}
    \caption{Left: the average approximation ratio of the top 20\% of solutions for Ising-20 found by VQE and our algorithm with different values of $r$ and different numbers of layers. Right: the empirical cumulative distribution function (ECDF) of the approximation ratios of the solutions found by the algorithms with the best number of layers.}
    \label{fig:small}
\end{figure}

The experiment demonstrates that the performance of our algorithm with \( r=20 \) is comparable to that of VQE. However, we also show in Fig.~\ref{fig:small} that using 20 qubits is inefficient, as our algorithm with \( r=1 \) and \( r=2 \), requiring only 5 and 8 qubits, respectively, provides better results. The poor performance of VQE and our algorithm with \( r=20 \) can be attributed to the following: both the auxiliary function (Eq.~(\ref{eq:e_cont})) and the VQE average energy (Eq.~(\ref{eq:vqe_energy_mod})) involve \( 2^{20} \) variables (for VQE we consider $P_{\mu}$ here). If these functions were optimized directly, the optimal solution would always be found, as all local minima map to the global minimum of the original problem. However, in the experiment we conducted, the quantum algorithms optimize only up to 320 parameters, leading to numerous local minima induced by the hardware-efficient ansatz \cite{larocca2024}, significantly deteriorating the quality of the solutions.

In contrast, the auxiliary function for \( r=1 \) and \( r=2 \) contains only 20 and 210 variables, respectively, but has local minima corresponding to $r$-local search. Here, the performance is constrained by the limitations of local search itself. Notably, the performance is also not very good -- the optimal solution to such a small problem is found in no more than 20\% of the runs. This limitation highlights why pure local search is seldom employed; instead, more advanced metaheuristics \cite{heuristics} that build upon local search are typically used to overcome these challenges. Our quantum algorithm can also be integrated into such sophisticated metaheuristics, providing a promising direction for future research to enhance solution quality.

\section{Extension to higher-order problems}
\label{sec:Appendix/HighOrder}
Here we show how to extend our method to the PUBO problem, which is equivalent to a generalized Ising model with higher-order spin interactions \cite{bybee2022efficient}. Let $n$ be the number of variables, $V=\{1,2,\ldots,n\}$. Then, the energy can be written as:

\begin{equation}
    E=\sum_{S\subseteq V} J_S \prod_{i\in S} Z_i,
\label{eq:high_order}
\end{equation}
with \(J_S \in \mathbb{R}\) the model parameters. Substituting Eq.~(\ref{eq:rev}) from the main text into Eq.~(\ref{eq:high_order}), we obtain:
\begin{equation}
    E=\sum_{S\subseteq V} J_S \prod_{i\in S} Z_i^{(0)}\prod_{k:i\in G_k} z_k.
\label{eq:high_order_neigh}
\end{equation}
The identity $z_k^2=1$ allows us to reduce the energy (Eq.~(\ref{eq:high_order_neigh})) to a multilinear form and then replace the discrete $z_k \in \{-1,1\}$ with continuous $q_k \in [-1,1]$. This results in an auxiliary function, $\mathcal{E}(\vec{q})$, that can be optimized with a variational quantum algorithm according to the same scheme as for the second-degree problem discussed in the main text.

\section{Efficient trivial mapping}
\label{sec:Appendix/Mapping}
Suppose the groups are ordered in the following sequence: \[\{0\},\{1\},\{2\}\ldots \{n-1\},\{0,1\},\{0,2\}\ldots\{0,n-1\},\{1,2\}\ldots \{n-2,n-1\},\{0,1,2\},\{0,1,3\}\ldots\]
Here, we outline an efficient algorithm for determining the group, $G$, corresponding to a given measurement output, $\mu$. This procedure enables the implementation of the trivial mapping described in Sec.~\ref{sec:Quantum/k_to_G} for a large number of groups where storing all groups explicitly becomes infeasible. The algorithm is a modified version of the lexicographic unranking method for combinations~\cite{Genitrini2021}.

\begin{enumerate}
    
    \item \textbf{Determine the Size of the Subset:}
    Let \( C(k) = \sum_{i=1}^k \binom{n}{i} \), where \( \binom{n}{i} \) is a binomial coefficient. The value \( C(k) \) represents the total number of subsets of size \( \leq k \). 
    
    Find the size, \( m \), of the subset, \( G \), such that:
    \[
    C(m-1) < \mu \leq C(m).
    \]
    Compute the adjusted position within subsets of size \( m \) as:
    \[
    \mu' = \mu - C(m-1).
    \]
    
    \item \textbf{Generate the \(\mu'\)-th Subset:}
    Start with an empty subset, \( G \), and construct the \( \mu'\)-th subset of size \( m \):
    \begin{enumerate}
        \item Initialize \( m_0 = m \) and \( \mu_0 = \mu' \).
        \item For each candidate element, \( x \), from \( 0 \) to \( n-1 \), do the following:
        \begin{itemize}
            \item Compute the number of subsets of size \( m_0 \) that can be formed with \( x \) as the smallest element:
            \[
            \text{count} = \binom{n-x-1}{m_0 - 1}.
            \]
            \item If \( \mu_0 \leq \text{count} \), add \( x \) to \( G \) and decrement \( m_0 \) by 1. Otherwise, subtract the \( \text{count} \) from \( \mu_0 \).
        \end{itemize}
        \item Repeat until \( m_0 = 0 \).
    \end{enumerate}
\end{enumerate}

This procedure requires $O(n)$ evaluations of large binomial coefficients. The largest possible binomial coefficient is $\binom{n}{r}$ which has an evaluation complexity $O[r M(d)]$, where $M(d)\leq O(d^2)$ is the complexity of $d$-bit number multiplication and $d\leq \log_2 n^r$. Therefore, the overall complexity of the procedure is $O[nr M(r \log n)]$. 

\section{Sampling the Most Probable Solutions}
\label{sec:Appendix/Sampling}
We propose an efficient algorithm to determine the \(S\) most probable configurations \(\vec{z}\) of a multivariate Bernoulli distribution with independent spins. The probability for the \(k\)-th spin to be down is defined as \(p_k := P(z_k = -1)\). The probability of a given configuration, \(\vec{z}\), is given by:

\[
P(\vec{z}) = \prod_{k: z_k = -1} p_k \prod_{k: z_k = 1} (1 - p_k).
\]

The goal is to compute the \(S\) most probable configurations, denoted as \(\vec{z}^{(1)}, \vec{z}^{(2)}, \ldots, \vec{z}^{(S)}\). The algorithm builds these solutions iteratively based on their probabilities.

The most probable configuration, \(\vec{z}^{(1)}\), can be straightforwardly determined as:
   \[
   z^{(1)}_k =
   \begin{cases}
       1 & \text{if } p_k < 0.5, \\
      -1 & \text{otherwise.}
   \end{cases}
   \]

Let \(\hat{F}_k \vec{z}\) denote the configuration obtained by flipping the \(k\)-th spin in \(\vec{z}\). If the \(k\)-th spin is in its most probable state, flipping it scales the probability of the configuration by a factor, \(g_k\), defined as:

\[
g_k = \frac{P(\hat{F}_k \vec{z})}{P(\vec{z})} = 
\begin{cases}
    p_k/(1 - p_k) & \text{if } p_k < 0.5, \\
    (1 - p_k)/p_k & \text{if } p_k \geq 0.5.
\end{cases}
\]
Since \(g_k \leq 1\), flipping a spin from its most probable state always reduces the overall probability of the configuration.

\begin{algorithm}
\caption{Sampling the Most Probable Configurations}
\label{alg:sampling}
\begin{algorithmic}[1]
\STATE \textbf{Input}: spin probabilities \(\vec{p}\), number of samples \(S\)
\STATE Compute \(g_k\) for each spin: \(g_k \gets p_k/(1 - p_k)\) if \(p_k < 0.5\), else \((1 - p_k)/p_k\)
\STATE Compute the initial configuration: \(\vec{z}^{(1)} \gets 1 - 2 \cdot \textsc{Round}(\vec{p})\)
\STATE Compute the probability of \(\vec{z}^{(1)}\): \(P^{(1)} \gets \prod_{k: z^{(1)}_k = -1} p_k \prod_{k: z^{(1)}_k = 1} (1 - p_k)\)
\STATE Initialize a priority list \(A \gets [(\vec{z}^{(1)}, P^{(1)})]\)
\FOR{\(k = 1\) to \(\textsc{Length}(\vec{p})\)}
    \STATE Create a new list \(B \gets []\)
    \FOR{\((\vec{z}, P)\) in \(A\)}
        \STATE Flip the \(k\)-th spin: \(\hat{F}_k \vec{z}\)
        \STATE Compute the new probability: \(P_{\text{new}} \gets g_k \cdot P\)
        \STATE Append \((\hat{F}_k \vec{z}, P_{\text{new}})\) to \(B\)
    \ENDFOR
    \STATE Merge \(A\) and \(B\): \(A \gets A \cup B\)
    \STATE Sort \(A\) by probability \(P\) in descending order
    \STATE Keep the top \(S\) configurations: \(A \gets A[:S]\)
\ENDFOR
\STATE \textbf{Return}: \(A\), containing the \(S\) most probable configurations and their probabilities
\end{algorithmic}
\end{algorithm}

The algorithm for sampling the most probable configurations is given in Alg.~\ref{alg:sampling}. It can be explained as follows:
\begin{enumerate}
\item \textbf{Initialization}: The algorithm starts with the most probable configuration \(\vec{z}^{(1)}\) and its probability \(P^{(1)}\). This configuration is stored in list $A$ along with its probability.
\item \textbf{Iterative Expansion}: At each step, the algorithm generates all the configurations obtained by flipping one spin of each current configuration in \(A\). The probabilities are updated using the precomputed \(g_k\).
\item \textbf{Pruning}: To prevent exponential growth in the number of configurations, $A$ is sorted by probability in descending order, and only the top $S$ configurations are retained. Importantly, ``child'' configurations of the rejected ones -- configurations that differ only in spins not yet considered for flipping -- are automatically excluded. Since $g_k\leq 1$ for all $k$, they have an even lower probability, which ensures that we don't lose any solutions of the interest.
\item \textbf{Termination}: After all the spins have been processed, \(A\) contains the \(S\) most probable configurations.
\end{enumerate}

The most time-consuming procedure in the algorithm is sorting the array $A$ containing at most $2S$ elements, which can be done in $O(S \log S)$. Therefore, the complexity of the algorithm is $O(lS\log S)$, where $l$ is the length of $\vec{p}$. However, if the circuit is measured with a finite number of shots, $N$, there are no more than $N$ nonzero $p_k$, reducing the complexity to $O(NS\log S)$.

\bibliography{references.bib}

\end{document}